%% file: vc.tex
\documentclass[11pt]{article}

\usepackage{amsmath,amssymb,amsfonts}
\usepackage{mathtools}
\usepackage{epsfig}
\usepackage{latexsym,nicefrac,bbm}
\usepackage{xspace}
\usepackage{color,fancybox,graphicx,url,subfigure}

\usepackage{times,fullpage}

\def\showauthornotes{0}

\def\showdraftbox{0}
\input{macros}

\newcommand{\eps}{\varepsilon}
\newcommand{\etal}{{\em et al.\ }}
\newcommand\calI{\mathcal{I}}
\newcommand\calV{\mathcal{V}}

\newcommand{\hvc}{\mbox{$k$-{\sc HypVC}}\xspace}
\newcommand{\hvcpartite}{\mbox{$k$-{\sc Hyp\-VC\--Par\-tite}}\xspace}

\newcommand{\VC}{\mbox{\sc{vc}}}
\newcommand{\LP}{\mbox{\sc{lp}}}

\title{Nearly Optimal NP-Hardness of Vertex Cover on
  $k$-Uniform $k$-Partite Hypergraphs
}
\author{Sushant Sachdeva
\thanks{
Department of Computer Science, Princeton University. \texttt{sachdeva@cs.princeton.edu} }
\and 
Rishi Saket \thanks{
Department of Computer Science, Princeton University. \texttt{rsaket@cs.princeton.edu}}
}

\begin{document}

\maketitle

%\draftbox

\begin{abstract}
  We study the problem of computing the minimum vertex cover on
  $k$-uniform $k$-partite hypergraphs when the $k$-partition is given.
  On bipartite graphs ($k=2$), the minimum vertex cover can be
  computed in polynomial time.  For general $k$, the problem was
  studied by Lov\'{a}sz~\cite{Lovasz75}, who gave a
  $\frac{k}{2}$-approximation based on the standard LP
  relaxation. Subsequent work by Aharoni, Holzman and
  Krivelevich~\cite{AHK96} showed a tight integrality gap of
  $\left(\frac{k}{2} - o(1)\right)$ for the LP relaxation. 
  While this problem was
  known to be NP-hard for $k\geq 3$, the first non-trivial NP-hardness
  of approximation factor of $\frac{k}{4}-\eps$ was shown in a recent
  work by Guruswami and Saket~\cite{GS10}. They also showed that
  assuming Khot's Unique Games Conjecture yields a $\frac{k}{2}-\eps$
  inapproximability for this problem, implying the optimality of
  Lov\'{a}sz's result.

In this work, we 
show that this problem is NP-hard to approximate within 
$\frac{k}{2}-1+\frac{1}{2k}-\eps$. This hardness factor 
is off from the optimal
by an additive constant of at most $1$ for $k\geq
4$. Our reduction relies on the \emph{Multi-Layered PCP} of
\cite{DGKR03} and uses a gadget -- based on
biased Long Codes -- adapted from the LP integrality gap of \cite{AHK96}. 
The nature of our reduction
requires the analysis of several Long Codes with different biases, for
which we prove structural properties of the so called
\emph{cross-intersecting} collections of set families -- 
variants of which have
been studied in extremal set theory.    

\end{abstract}

\section{Introduction}
A $k$-uniform hypergraph $G=(V,E)$ consists of a set of vertices $V$
and a collection of hyperedges $E$ such that each hyperedge contains
exactly $k$ vertices. A vertex cover for $G$ is a subset of vertices
$\calV \subseteq V$ such that every hyperedge $e$ contains at least
one vertex from $\calV$ i.e. $e \cap \calV \neq
\emptyset$. Equivalently, a vertex cover is a hitting set for the
collection of hyperedges $E$. The complement of a vertex cover is
called an \emph{Independent Set}, which is a subset of
vertices $\calI$ such that no hyperedge $e \in E$ is contained inside
$\calI$ i.e. $e \nsubseteq \calI$.

The \hvc  problem is to compute the minimum vertex cover in a
$k$-uniform hypergraph $G$. It is an extremely well studied
combinatorial optimization problem, especially on graphs ($k=2$), 
and is known to be NP-hard. Indeed, the minimum vertex cover problem on graphs 
was
one of Karp's original 21 NP-complete problems~\cite{Karp72}. On the other
hand, the simple greedy algorithm that picks a maximal collection of
disjoint hyperedges and includes all vertices in the edges in the
vertex cover gives a $k$-approximation, which is also obtained by the
standard LP relaxation of the problem. The best algorithms known
today achieve only a marginally better approximation factor of 
$(1-o(1))k$~\cite{Karakostas05,Halperin02}. 

On the intractability side, there have been several results. For the
case $k=2$, Dinur and Safra~\cite{DS02} obtained an NP-hardness of 
approximation factor of
1.36, improving on a $\frac{7}{6}-\eps$ hardness by
H\r{a}stad~\cite{Hastad01}. For general $k$ a sequence of successive
works yielded improved NP-hardness 
factors: $\Omega(k^{1/19})$ by 
Trevisan~\cite{Trevisan01};
$\Omega(k^{1-\eps})$ by Holmerin~\cite{Holmerin02}; $k-3-\eps$ by Dinur,
Guruswami and Khot~\cite{DGK02}; and the currently best
$k-1-\epsilon$ due to Dinur, Guruswami, Khot and
Regev~\cite{DGKR03}. In \cite{DGKR03}, the authors build upon \cite{DGK02}
and the work of Dinur and Safra \cite{DS02}. 
Moreover, assuming Khot's
Unique Games Conjecture (UGC)~\cite{Khot02}, Khot and Regev~\cite{KR08} 
showed an essentially
optimal $k-\eps$ inapproximability. This result was further
strengthened in different directions by Austrin, Khot and
Safra~\cite{AKS09}
and by Bansal and Khot~\cite{BK10}.

\smallskip
{\it \large Vertex Cover on $k$-uniform $k$-partite
Hypergraphs}

In this paper we study the minimum vertex problem on $k$-partite $k$-uniform
hypergraphs, when the underlying partition is given. We denote this
problem as \hvcpartite. This is an interesting
problem in itself and its variants have been studied for
applications related to databases such as distributed data
mining~\cite{FOMPT03}, schema mapping discovery~\cite{GS09} and 
optimization of finite automata~\cite{ISY05}. On bipartite graphs ($k=2$), by 
K\"{o}enig's Theorem computing the minimum vertex cover 
is equivalent to computing the maximum matching which can be done
efficiently.  
For general $k$, the problem was studied 
by Lov\'{a}sz who, in his doctoral thesis~\cite{Lovasz75}, 
proved the following upper bound. 
\begin{Thm}[Lov\'{a}sz~\cite{Lovasz75}]
For every $k$-partite $k$-uniform hypergraph $G$:
$\VC(G)/\LP(G) \le k/2$,
where $\VC(G)$ denotes the size of the minimum vertex cover and $\LP(G)$
denotes the value of the standard LP relaxation. This yields an
efficient $k/2$ approximation for $\hvcpartite$. 
\end{Thm}
The above upper bound was shown to be tight
by Aharoni, Holzman and Krivelevich~\cite{AHK96} who proved the
following theorem.
\begin{Thm}[Aharoni et al.\cite{AHK96}]
For every $k \ge 3$, there exists a family of 
$k$-partite $k$-uniform hypergraphs $G$
such that $\VC(G)/\LP(G) \ge k/2-o(1)$. Thus, the integrality gap of
the standard LP relaxation is $k/2 - o(1)$.
\end{Thm}A proof of the above theorem describing the integrality gap
construction is included in Section \ref{sec:aharoni}.
The problem was shown to be
APX-hard in \cite{ISY05} and \cite{GS09} for $k=3$ which can be
extended easily to $k \geq 3$. 
A recent work of Guruswami and Saket~\cite{GS10}
showed the following non-trivial hardness of approximation 
factor for general $k$.
\begin{Thm}[Guruswami and Saket~\cite{GS10}]
For any $\epsilon > 0$ and $k \geq 5$, \hvcpartite is NP-hard to approximate
within a factor of $\frac{k}{4}-\epsilon$. Assuming the UGC yields an
optimal hardness factor of 
$\frac{k}{2}-\epsilon$ for $k\geq 3$.
\end{Thm}
{\bf Our Contribution.} We show a nearly optimal
NP-hardness result for approximating \hvcpartite.
\begin{Thm}\label{thm:main}
For any $\epsilon > 0$ and integer $k \ge 4$, it is NP-hard to
approximate the minimum vertex cover on $k$-partite $k$-uniform
hypergraphs within to a factor of $\frac{k}{2}-1+\frac{1}{2k}-\epsilon$.
\end{Thm}
Our result significantly improves on the NP-hardness factor obtained
in \cite{GS10} and is off by at most an additive constant of $1$ from
the optimal for any $k \geq 4$. The next few paragraphs give an
overview of the techniques used in this work.

\smallskip
{\bf Techniques.} 
\iffalse In the last two decades, the PCP Theorem~\cite{AS, ALMSS}
has been used to show hardness of approximation for several optimization
problems. Many of these results are based on reductions from Label Cover
-- a two variable constraint satisfaction problem (CSP) 
obtained by combining Raz's Parallel 
Repetition Theorem~\cite{Raz} with the PCP Theorem.
However, as noted in \cite{DGKR03}, a reduction from the
standard bipartite version of Label Cover is unlikely to yield a
result close to optimal for the Minimum Vertex Cover problem on general
$k$-uniform hypergraphs. 
\fi
It is helpful to first briefly review the hardness reduction of
\cite{DGKR03} for \hvc which begins with the construction of a new 
{\it Multi-Layered
PCP}. This is a two variable CSP consisting of 
several {\it layers} of variables, and constraints between the 
variables
of each pair of layers. The work of \cite{DGKR03} shows that it is NP-hard 
to find a labeling to the variables which satisfies a small fraction
of the constraints between {\it any} two layers, even if there is a
labeling that satisfies all the constraints of the instance. The 
reduction to a $k$-uniform hypergraph
(as an instance of \hvc) involves replacing each
variable of the PCP with a biased Long Code, defined in~\cite{DS02}, 
where the bias depends on $k$. 

The starting point for our hardness reduction for \hvcpartite is -- as
in \cite{DGKR03} -- the
Multi-Layered PCP. While we do not explicitly
construct a standalone Long Code based gadget, our reduction can be
thought of as
adapting the
integrality gap construction of Aharoni \etal~\cite{AHK96} into a Long
Code based gadget in a manner that preserves the $k$-uniformity and
$k$-partiteness of the integrality gap.

Such transformations of integrality gaps into Long
Code based gadgets have recently been studied in the works of
Raghavendra~\cite{R08} and Kumar, Manokaran, Tulsiani and Vishnoi~\cite{KMTV09} which show this for a wide class of CSPs
and their 
appropriate LP and SDP
integrality gaps. These Long Code based gadgets can 
be combined with a Unique Games instance to yield tight UGC based
hardness results, where the reduction is analyzed via the Mossel's {\it
Invariance Principle}~\cite{Mossel08}. 
Indeed, for \hvcpartite the work of Guruswami and Saket~\cite{GS10}
combines the integrality gap of \cite{AHK96} with (a slight
modification) of the approach of Kumar \etal~\cite{KMTV09} to obtain 
an optimal UGC based
hardness result. 

Our reduction, on the other hand, combines Long Codes with the
Multi-Layered PCP instead of Unique Games and so we cannot adopt a
Invariance Principle based analysis. 
Thus, in a flavor similar to that of \cite{DGKR03}, 
our analysis is via extremal combinatorics. However, our {\it gadget}
involves several biased Long Codes with different biases 
and each hyperedge includes vertices from different Long Codes, unlike
the construction in \cite{DGKR03}. For 
our analysis, we use 
structural properties of a
{\it cross-intersecting} collection of set families. A collection of set
families
is {\it cross-intersecting} if any intersection of subsets -- 
each chosen from a
different family -- is large. Variants of this notion
have previously been studied in extremal set theory, see 
for example~\cite{AL09}. 
We prove an upper bound on the measure of the 
smallest family in such a collection. 
This enables a
small vertex cover (in the hypergraph of our reduction) to be 
{\it decoded} into a good labeling to the
Multi-Layered PCP. 

The next section defines and analyzes the above mentioned 
cross-intersecting set families. 
Section \ref{sec:multi} defines the Multi-Layered PCP 
of Dinur \etal \cite{DGKR03} and states their hardness for it. 
In Section \ref{sec:redn} we describe our reduction and prove
Theorem \ref{thm:main}.
%Our construction relies on the biased Long-Code that was defined in
%\cite{DS02} and then successfully used in several other reductions
%\cite{KR,DGKR,KMTV}. The biased Long-Code for a set $[n]$ has one
%vertex for each subset of $[n]$ and the measure of a subset $F$ is
%defined as $\mu_p(F) = p^{|F|}(1-p)^{n-|F|}$, for some fixed parameter
%$p$. 
%From the works of \cite{Rag,KMTV}, it is clear that for large classes
%of \emph{Constraint Satisfaction Problems}, the Long-Code gadgets
%required for reductions from Unique Games can be constructed from an
%integrality gap instance for a natural LP or SDP. For reductions from
%the Raz verifier, it's less obvious how to translate integrality gaps
%into gadgets. Nevertheless, we manage to modify the LP integrality gap
%of Aharoni et al. \cite{Aharoni} and use it to construct a gadget that
%works.
%Like some of the other vertex cover reductions \cite{DS02,KR,DGKR}, our
%analysis has some flavor of extremal combinatorics. We need to analyze
%\emph{Cross-Intersecting families}. A collection of families $\calF_i
%\subseteq 2^{[n]}$ is called $t$-cross-intersecting if for every choice of
%$F_i \in \calF_i$, the intersection $\cap_{i} F_i$ has at least $t$
%elements. We show that amongst a collection of families that is
%$t$-cross-intersecting, at least one of them must have small measure
%under the biased measure $\mu_p$ (with the right choice of $p$). We
%believe this result might be interesting on its own since we are
%required to consider families under measures $\mu_p$ with different
%values of $p$ (as long as the sum of biases is bounded above).
\section{Cross-Intersecting Set Families}
\label{sec:intersecting}
We use the
notation $[n] =\{1,\ldots,n\}$ and $2^{[n]} = \{F\ |\ F \subseteq
[n]\}$. We begin by defining cross-intersecting set families:
\begin{Def}
A collection of $k$ families $\calF_1,\ldots,\calF_k \subseteq 2^{[n]}$,
is called $k$-wise $t$-cross-intersecting if for every choice of sets $F_i \in  \calF_i$
for  $i = 1,\ldots,k$, we have $|F_1\cap\ldots\cap F_k| \ge t$.
\end{Def}

We will work with the $p$-biased measure on the subsets of $[n]$, which is defined as follows: 
\begin{Def}
Given a bias parameter $0 < p < 1$, we define the
measure $\mu_p$ on the subsets of $[n]$ as: $
\mu_p(F) ~:=~ p^{|F|}\cdot(1-p)^{n - |F|}\ .$
The measure of a family $\calF$ is defined as $\mu_p(\calF) =
\sum_{F \in \calF} \mu_p(F)$.
\end{Def}

Now, we introduce an important technique for analyzing
cross-intersecting families -- the shift operation (see Def 4.1,
pg. 1298 \cite{Handbook}). Given a family
$\calF$, define the $(i,j)$-shift as follows:
\begin{align*}
S_{ij}^{\calF}(F) = \left \{ 
\begin{array}{ll}
(F \cup \{i\} \backslash \{j\})  & \textrm{if } j \in F,\ i \notin F \textrm{ and } (F \cup\{i\} \backslash \{j\}) \notin \calF \\
F & \textrm{otherwise.}
\end{array}
\right .
\end{align*}
Let the $(i,j)$-shift of a family $\calF$ be
$S_{ij}(\calF) = \{S_{ij}^{\calF}(F)\ |\ F \in \calF\}$. Given a family
$\calF \subseteq 2^{[n]}$, we repeatedly apply $(i,j)$-shift for $1 \le i <
j \le n$ to $\calF$ until we obtain a family that is invariant under
these shifts. Such a family is called a \emph{left-shifted family} and
we will denote it by $S(\calF)$. 

The following observations about left-shifted families 
follow from the definition.
\begin{Obs}
\label{obs:shift}
Let $\calF \subseteq 2^{[n]}$ be a left-shifted family. Consider $F \in \calF$ such that $i \notin F$ and $j
\in F$ where $i < j$. Then, $(F \cup \{i\}
\backslash \{j\})$ must be in $\calF$.
\end{Obs}

\begin{Obs}
\label{obs:shift-measure}
Given $\calF \subseteq 2^{[n]}$, there is a bijection between the sets in $\calF$ and $S(\calF)$ that
preserves the size of the set. Thus, for any fixed $p$, the measures of $\calF$ and $S(\calF)$ are the same under $\mu_p$ i.e. $\mu_p(\calF) = \mu_p(S(\calF))$. 
\end{Obs}

The following lemma shows that the cross-intersecting
property is preserved under left-shifting.
\begin{Lem}
\label{lem:shift-intersecting}
Consider families $\calF_1,\ldots,\calF_k \subseteq 2^{[n]}$ that are
$k$-wise $t$-cross-intersecting. Then, the families $S(\calF_1),\ldots,S(\calF_k)$ are also
$k$-wise $t$-cross-intersecting.
\end{Lem}
\begin{proof}
Given the assumption, we will prove that
$S_{ij}(\calF_1),\ldots,S_{ij}(\calF_k)$ are $k$-wise
$t$-cross-intersecting. A simple induction would then imply the
statement of the lemma.

Consider arbitrary sets $F_i \in \calF_i$. By our assumption, $|F_1 \cap \ldots
\cap F_k| \ge t$. It suffices to prove that $|S_{ij}^{\calF_1}(F_1)
\cap \ldots \cap S_{ij}^{\calF_k}(F_k)| \ge t$. If $j \notin F_1 \cap
\ldots \cap F_k$, the claim is true since the only element being
deleted is $j$. Thus, for all $l \in [k]$, $j \in F_k$. If for all $l \in
[k]$, $S_{ij}^{\calF_l}(F_l) = F_l$, the claim is trivial. Thus, let
us assume wlog that $S_{ij}^{\calF_1}(F_1) \neq F_1$. Thus, $i \notin
F_1$ and hence $i \notin F_1 \cap \ldots \cap F_k$. Now, if $i \in
S_{ij}^{\calF_1}(F_1) \cap \ldots \cap S_{ij}^{\calF_k}(F_k)$, we get
that $j$ is replaced by $i$ in the intersection and we are done. Thus,
we can assume wlog that $i \notin S_{ij}^{\calF_2}(F_2)$. This implies
that $i \notin F_2$ and $F_2 \cup \{i\} \backslash \{j\} \in \calF_2$. Now
consider $F_1 \cap (F_2 \cup \{i\} \backslash \{j\}) \cap F_3 \cap \ldots
\cap F_k$. Since we are picking one set from each
$\calF_i$, it must have at least $t$ elements, but this intersection
does not contain $j$ and hence it is a subset of
$S_{ij}^{\calF_1}(F_1) \cap \ldots \cap S_{ij}^{\calF_k}(F_k)$,
implying that $|S_{ij}^{\calF_1}(F_1)
\cap \ldots \cap S_{ij}^{\calF_k}(F_k)| \ge t$.
\end{proof}

Next, we prove a key structural lemma about cross-intersecting
families which states that for at least one of the families, all of its subsets have a dense prefix.
% The proof of the following theorem can
% be better understood in terms of arrangements of balls in
% bins. Imagine $n$ bins, each with a label from $1$ to $n$. Consider $F
% \subseteq [n]$. For each
% element $x \in [n]\backslash F$, we throw a ball into the bin labeled
% $x$. Given any set $F$, we can construct such an arrangement of balls in
% bins. Also, if we are provided an arrangement of balls in the $n$ bins
% such that there is at most one ball in each bin, we can recover the
% set it represents by defining $F^c$ to be the set of bins that contain
% a ball.
\begin{Lem}
\label{lem:intersecting-structure}
Let $q_1,\ldots,q_k \in (0,1)$ be $k$ numbers such that $\sum_i q_i \ge 1$ 
and let $\calF_1,\ldots,\calF_k \subseteq
2^{[n]}$ be left-shifted families that
are $k$-wise $t$-cross-intersecting for some $t \ge 1$. Then, there exists a $j
\in [k]$ such that for all sets $F
\in \calF_j$, there exists a positive integer $r_F \le n-t$ such that $|F\cap [t+r_F]| > (1-q_i)(t+r_F)$.
\end{Lem}
\begin{proof}
Let us assume to the contrary that for every $i \in [k]$, there exists a set
$F_i \in \calF_i$ such that for all $r \ge 0$, $|F_i \cap [t+r]| \le
(1-q_i)(t+r)$. The following combinatorial argument shows that the families $\calF_i$
cannot be $k$-wise $t$-cross-intersecting.

Let us construct an arrangement of balls
and bins where each ball is colored with one of $k$ colors. Create $n$ bins
labeled $1,\ldots,n$. For each $i$ and for every $x
\in [n] \backslash F_i$, we place a ball with color $i$ in the bin
labeled $x$. Note that a bin can have several balls, but they
must have distinct colors. Given such an arrangement, we can recover the
sets it represents by defining $F_i^c$ to be the set of bins that contain
a ball with color $i$. 

Our initial assumption implies that $|F_i^c \cap [t+r]| \ge q_i(t+r)$. Thus, there are at least $\ceil{q_i(t+r)}$ balls with color $i$ in bins labeled
$1,\ldots,t+r$. The total number of balls in
bins labeled $1,\ldots,t+r$ is,
\[\sum_{i=1}^k |F_i^c \cap [t+r]| ~\ge~ \sum_{i=1}^k \ceil{q_i(t+r)} ~\ge~ \sum_{i=1}^k q_i(t+r) 
~\ge~ (t+r) ~\ge~ r+1,\]
where the last two inequalities follow using $\sum_i q_i
\ge 1$ and $t \ge 1$. 

Next, we describe a procedure to manipulate the above arrangement of balls.
\begin{program}
\FOR \= $r$ := 0 to $n-t$ \\
   \>    \IF bin $t+r$ is empty \\
   \>    \THEN  \IF a bin labeled from $1$ to $t-1$ contains a ball \THEN move it to bin $t+r$ \\
\> \>\> \>\ELSE \IF a bin labeled from $t$ to $t+r-1$ contains two
balls \THEN move one of them to bin $t+r$ \\
\> \>\>\>\>\> \ELSE output ``error''
\end{program}

We need the following lemma.
\begin{Lem} The above procedure satisfies the following properties:\\
1. The procedure never outputs \emph{error}.\\
2. At every step, any two balls in the same bin have different
colors. \\
3. At step $r$, define $G_i^{(r)}$ to be the set of labels of
the bins that do not contain a ball of color $i$. Then, for all $i \in
[k]$, $G_i^{(r)} \in \calF_i$. \\
4. After step $r$, the bins $t$ to $t+r$ have at least one ball each.
\end{Lem}
\begin{proof}
1. If it outputs error at step $r$, there must be at
most $r-1$ balls in bins $1$ to $t+r$. This is false at $r =
0$. Moreover, at step $r^\prime < r$, we could
have moved a ball only to a bin labeled in $[t,t+r]$.  
Thus, we get a contradiction. \\
2. Note that this is true at $r=0$ and a ball is only moved to an empty bin, which proves the claim. \\
3. Whenever we move
a ball from bin $i$ to $j$, we have $i < j$. Since 
$\calF_i$ are left-shifted, by repeated
application of Observation \ref{obs:shift}, we get that at step $r$,
$G_i^{(r)} \in \calF_i$. \\
4. Since the procedure never outputs error, at step $r$, if the bin
$t+r$ is empty, the procedure places a ball in it while not emptying
any bin labeled between $[t,t+r-1]$. This proves the claim.
\end{proof}
The above lemma implies that at the end of the procedure (after $r=n-t$), there is a ball in each of the bins
labeled from $[t,n]$. Thus, the sets $G_i = G_i^{(n-t)}$ satisfy $\cap_i G_i \subseteq
[t-1]$ and hence $|\cap_i G_i| \le t-1$. Also, we know that $G_i \in
\calF_i$. Thus, the families $\calF_i$ cannot be $k$-wise
$t$-cross-intersecting. This completes the proof of Lemma \ref{lem:intersecting-structure}.
\end{proof}

The above lemma, along with a Chernoff bound argument, shows that: Given
a collection of $k$-wise $t$-cross-intersecting families, one of them
must have a small
measure under an appropriately chosen bias.

\begin{Lem}
\label{lem:intersecting-measure}
For arbitrary $\epsilon,\delta > 0$, there exists some $t
= O\left(\frac{1}{\delta^2}\left(\log
  \frac{1}{\epsilon} + \log
  \left(1+\frac{1}{2\delta^2}\right)\right)\right)$ such that the
following holds: 
Given $k$ numbers $0 < q_i < 1$ such that $\sum_i q_i \ge 1$ 
and $k$ families, $\calF_1,\ldots,\calF_k \subseteq 2^{[n]}$, that
are $k$-wise $t$-cross-intersecting, there exists a $j$ such that
$\mu_{1-q_i-\delta}(\calF) < \epsilon$.
\end{Lem}
\begin{proof}First we prove the following lemma derived from 
the Chernoff bound.
\begin{Lem}
\label{lem:chernoff}
For arbitrary $\epsilon,\delta > 0$ and $0 < q < 1$, there exists some $t
= O\left(\frac{1}{\delta^2}\left(\log
  \frac{1}{\epsilon} + \log
  \left(1+\frac{1}{2\delta^2}\right)\right)\right)$ such that the following holds:

Any family $\calF \subseteq 2^{[n]}$ that satisfies that for every $F \in
\calF$, there exists an integer $r_F \ge 0$ such that $|F \cap [t+r_F]|
\ge (1-q)(t+r_F)$ must have $\mu_{1-q-\delta}(\calF) < \epsilon$.  
\end{Lem}
\begin{proof}
Note that $\mu_{1-q-\delta}(\calF)$ is equal to the probability that for a
random set $F$ chosen according to $\mu_{1-q-\delta}$ lies in
$\calF$. Thus, $\mu_{1-q-\delta}(\calF)$ is bounded by the probability that for a
random set $F$ chosen according to $\mu_{1-q-\delta}$, there exists
an $r_F$ that satisfies $|F \cap [t+r_F]| \ge (1-q)(t+r_F)$.

The Chernoff bound states that for a set of $m$ independent bernoulli
random variables $X_i$, with $\Pr[X_i = 1] = 1-q-\tau$,
\[\Pr\left[\sum_{i=1}^m X_i  \ge (1-q)m\right] \le e^{-2m\tau^2}\]

Thus, we get that for any $r \ge 0$, $\Pr[|F \cap [t+r]| \ge
(1-q)(t+r)] \le e^{-2(t+r)\delta^2}$. Summing over all $r$, we
get that,
\[\mu_{1-q-\delta}(\calF) \le \sum_{r \ge 0}e^{-2(t+r)\delta^2} \le
\frac{e^{-2t\delta^2}}{1-e^{-2\delta^2}} \le
e^{-2t\delta^2}\left(1+\frac{1}{2\delta^2}\right).\]

Thus, for $t = \Omega\left(\frac{1}{\delta^2}\left(\log
  \frac{1}{\epsilon} + \log
  \left(1+\frac{1}{2\delta^2}\right)\right)\right)$, $\mu_{1-q-\delta}(\calF)$
will be smaller than $\epsilon$.
\end{proof}
We now continue with the proof of Lemma \ref{lem:intersecting-measure}.
Our $t$ will be dictated by Lemma \ref{lem:chernoff} and will be
decided later. Consider the left-shifted families $S(\calF_i)$. By Lemma
\ref{lem:shift-intersecting}, we get that these families are also
$k$-wise $t$-cross-intersecting. Now, we can apply Lemma 
\ref{lem:intersecting-structure} with the given $q_i$'s to conclude
that there must exist a $j$ such that for all sets $F \in S(\calF_j)$,
there exists an $r_F$ such that $|F \cap [t+r_F]| > (1-q_j)(t+r_F)$. 

Now, we can use Lemma \ref{lem:chernoff} to conclude that if $t$ is
large enough ($t=\Omega\left(\frac{1}{\delta^2}\left(\log
  \frac{1}{\epsilon} + \log
  \left(1+\frac{1}{2\delta^2}\right)\right)\right)$ suffices), then
$S(\calF_j)$ must have measure at most $\epsilon$ under the measure
$\mu_{1-q_j-\delta}$, but this along with Observation
\ref{obs:shift-measure} implies that $\mu_{1-q_j-\delta}(\calF_j) < \epsilon$.
\end{proof}

\input{construction.tex}
\bibliographystyle{plain}
\bibliography{refs-vc}

\appendix
\section{LP Integrality Gap for \hvcpartite}\label{sec:aharoni}
\label{sec-aharoni}
This section describes the $\frac{k}{2} - o(1)$ 
integrality gap construction of Aharoni
\etal \cite{AHK96} for the standard LP relaxation for \hvcpartite. 
The hypergraph that is constructed is
unweighted.

Let $r$ be a (large) positive integer. The vertex set $V$ of the hypergraph 
is partitioned into subsets $V_1, \dots, V_k$ where, for all $i=
1,\dots, k$,
\begin{equation}
V_i = \{x_{ij}\ \mid\ j = 1,\dots, r\}\cup\{y_{il}\ \mid\ l = 1,\dots,
rk+1\}.
\end{equation}
Before we define the hyperedges, for convenience we shall define the
LP solution.
The LP values of the vertices are as given by the function $h : V
\mapsto [0,1]$ as follows: for all $i = 1, \dots, k$,
\begin{eqnarray}
h(x_{ij}) = \frac{2j}{rk}, && \forall j =1, \dots, r  \nonumber \\
h(y_{il}) = 0, && \forall l =1, \dots, rk+1.\nonumber
\end{eqnarray}
The set of hyperedges is naturally defined to be the set of all possible 
hyperedges, choosing exactly one vertex from each $V_i$ such that the
sum of the LP values of the corresponding vertices is at least $1$.
Formally,
\begin{equation}
E = \{e \subseteq V\ \mid\ \forall i\in [k],\ |e\cap V_i| =1\textnormal{
and }\sum_{v \in e} h(v) \geq 1\}.
\end{equation} 
Clearly the graph is $k$-uniform and $k$-partite with 
$\{V_i\}_{i\in[k]}$ being the $k$-partition of $V$.

The value of the LP solution is 
\begin{equation}
\sum_{v\in V}h(v) = k\sum_{j\in [r]}\frac{2j}{rk} = r + 1.
\end{equation}

Now let $V'$ be a minimum vertex cover in the hypergraph.
To lower bound the size of the minimum vertex cover, we first note
that the set $\{v \in V\ \mid\ h(v) > 0\}$ is a vertex cover of size
$rk$, and therefore $|V'| \leq rk$. 
Also, for any $i \in [k]$ the vertices $\{y_{il}\}_{l\in[rk+1]}$
have the same neighborhood. Therefore, we can assume that $V'$ 
has no vertex $y_{il}$, otherwise it will
contain at least $rk+1$ such vertices. 

For all $i \in [k]$ let define indices $j_i\in [r]\cup\{0\}$ as
follows:
\begin{equation}
j_i = \begin{cases}
	0\textnormal{ \ \ \ if: } \forall j\in[r],\ x_{ij}\in V', \\
	\max\ \{j\in[r]\ \mid\ x_{ij}\not\in V'\}\textnormal{ \
otherwise}. 
\end{cases}
\end{equation}
It is easy to see that since $V'$ is a vertex cover, 
$$\sum_{i\in [k]}h(x_{ij_i}) < 1,$$
which implies,
$$\sum_{i\in[k]}j_i < \frac{rk}{2}.$$
Also, the size of $V'$ is lower bounded by $\sum_{i\in[k]}(r - j_i)$.
Therefore,
\begin{equation}
|V'| \geq  \sum_{i\in[k]}(r - j_i) \geq rk - \sum_{i\in[k]}j_i \geq rk
- \frac{rk}{2} = \frac{rk}{2}.
\end{equation}
The above combined with the value of the LP solution yields an
integrality gap of $\frac{rk}{2(r+1)} \geq \frac{k}{2} - o(1)$ for
large enough $r$.

\end{document}

%% file: macros.tex
\newcommand{\marginlabel}[1]%
{\mbox{}\marginpar{\it{\raggedleft\hspace{0pt}#1}}}
\newcommand{\ceil}[1]{\left\lceil\, {#1}\,\right\rceil}

\definecolor{Mygray}{gray}{0.8}

 \ifcsname ifcommentflag\endcsname\else
  \expandafter\let\csname ifcommentflag\expandafter\endcsname
                  \csname iffalse\endcsname
\fi

\ifnum\showauthornotes=1

\else

\fi

\ifnum\showauthornotes=1
\newcommand{\Authornote}[2]{{\sf\small\color{red}{[#1: #2]}}}
\newcommand{\Authoredit}[2]{{\sf\small\color{red}{[#1]}\color{blue}{#2}}}
\newcommand{\Authorcomment}[2]{{\sf \small\color{gray}{[#1: #2]}}}
\newcommand{\Authorfnote}[2]{\footnote{\color{red}{#1: #2}}}
\newcommand{\Authorfixme}[1]{\Authornote{#1}{\textbf{??}}}
\newcommand{\Authormarginmark}[1]{\marginpar{\textcolor{red}{\fbox{%\Large
#1:!}}}}
\else
\newcommand{\Authornote}[2]{}
\newcommand{\Authoredit}[2]{}
\newcommand{\Authorcomment}[2]{}
\newcommand{\Authorfnote}[2]{}
\newcommand{\Authorfixme}[1]{}
\newcommand{\Authormarginmark}[1]{}
\fi

\newcommand\calF{\mathcal{F}}

\newcommand\calH{\mathcal{H}}

\newlength{\pgmtab}  %  \pgmtab is the width of each tab in the
 \newenvironment{program}{%
\begin{tabbing}\hspace{0em}\=\hspace{0em}\=%
\hspace{\pgmtab}\=\hspace{\pgmtab}\=\hspace{\pgmtab}\=\hspace{\pgmtab}\=%
\hspace{\pgmtab}\=\hspace{\pgmtab}\=\hspace{\pgmtab}\=\hspace{\pgmtab}\=%
\+\+\kill}{\end{tabbing}}

\newcommand {\ELSE}{{\bf else\ }}
\newcommand {\IF}{{\bf if\ }}
\newcommand {\FOR}{{\bf for\ }}

\newcommand {\THEN}{\mbox{\bf then\ }}

\newtheorem{Thm}{Theorem}[section]
\newtheorem{Lem}[Thm]{Lemma}

\newtheorem{Def}[Thm]{Definition}

\newtheorem{Obs}[Thm]{Observation}

 {
	\begin{enumerate}}{\end{enumerate}}

\def\FullBox{\hbox{\vrule width 6pt height 6pt depth 0pt}}

\def\qed{\ifmmode\qquad\FullBox\else{\unskip\nobreak\hfil
\penalty50\hskip1em\null\nobreak\hfil\FullBox
\parfillskip=0pt\finalhyphendemerits=0\endgraf}\fi}

\def\qedsketch{\ifmmode\Box\else{\unskip\nobreak\hfil
\penalty50\hskip1em\null\nobreak\hfil$\Box$
\parfillskip=0pt\finalhyphendemerits=0\endgraf}\fi}

\newenvironment{proof}{\begin{trivlist} \item {\bf Proof:~~}}
   {\qed\end{trivlist}}

\newcounter{lecnum}

\newlength{\tpush}
\ifnum\showdraftbox=1

\else

\fi

%% file: construction.tex
\section{Multi-Layered PCP}\label{sec:multi}
In this section we describe the Multi-Layered PCP constructed in
\cite{DGKR03} and its useful properties. An instance $\Phi$ of the
Multi-Layered PCP is parametrized by
integers $L, R > 1$. The PCP consists of $L$ sets of variables $X_1,
\dots, X_L$. The label set (or range) of the variables in the $l^\textrm{th}$
set $X_l$ is a set $R_{X_l}$ where $|R_{X_l}| = R^{O(L)}$. For any two
integers $1 \leq l < l' \leq L$, the PCP has a set of constraints
$\Phi_{l,l'}$ in which each constraint depends on one variable $x \in
X_l$ and one variable $x' \in X_{l'}$. The constraint (if it exists)
between $x \in X_l$ and $x' \in X_{l'}$ ($l < l'$) is denoted and
characterized by a projection 
$\pi_{x\rightarrow x'} : R_{X_l}\mapsto R_{X_{l'}}$. A labeling to $x$
and $x'$ satisfies the constraint $\pi_{x\rightarrow x'}$ if the 
projection (via $\pi_{x\rightarrow x'}$) of the label
assigned to $x$ coincides with the label assigned to $x'$. 

The following useful `weak-density' property of the Multi-Layered PCP was
defined in \cite{DGKR03}.
\begin{Def}\label{def-weakly-dense}
An instance $\Phi$ of the Multi-Layered PCP with $L$
layers is \textnormal{weakly-dense} if for any $\delta > 0$, given $m
\geq \lceil\frac{2}{\delta}\rceil$ layers $l_1 < l_2 < \dots < l_m$
and given any sets $S_i\subseteq X_{l_i}$, for $i\in[m]$ such that
$|S_i|\geq \delta|X_{l_i}|$; there always exist two layers
$l_{i'}$ and $l_{i''}$ such that the constraints between the variables
in the sets $S_{i'}$ and $S_{i''}$ is at least
$\frac{\delta^2}{4}$ fraction of the constraints between the sets 
$X_{l_{i'}}$ and $X_{l_{l''}}$.
\end{Def}

The following inapproximability of the Multi-Layered PCP was proven by
Dinur et al. \cite{DGKR03} based on the PCP Theorem (\cite{AS},
\cite{ALMSS}) and Raz's Parallel Repetition Theorem (\cite{Raz}).
\begin{Thm}\label{thm-multi}
There exists a universal constant $\gamma >0$ such that for any
parameters $L > 1$ and $R$, there is a weakly-dense $L$-layered PCP
$\Phi = \cup \Phi_{l,l'}$ such that it is NP-hard to distinguish
between the following two cases:
\begin{itemize}
\item \textnormal{{\bf YES} Case:} There exists an assignment of
labels to the
variables of $\Phi$ that satisfies all the constraints.
\item \textnormal{{\bf NO} Case:} For every $1\leq l < l' \leq L$, 
not more that
$1/R^\gamma$ fraction of the constraints in $\Phi_{l,l'}$ can be
satisfied by any assignment. 
\end{itemize}
\end{Thm}

\section{Hardness Reduction for {\sc HypVC-Partite}\label{sec:redn}}
\subsection{Construction of the Hypergraph}
\label{sec:construction}
Fix a $k \geq 3$, an arbitrarily small parameter $\eps > 0$ and
let $r = \lceil 10\eps^{-2}\rceil$. We shall construct a $(k+1)$-uniform
$(k+1)$-partite hypergraph as an instance of $(k+1)$-{\sc HypVC-Partite}.
Our construction will be a reduction from an instance $\Phi$ of the
Multi-Layered PCP with number of layers $L = 32\eps^{-2}$ and
parameter $R$ which shall be chosen later to be large enough. It involves 
creating, for each variable of the PCP,
several copies of the Long Code endowed with different
biased measures as explained below.

Over any domain $T$, a Long Code $\calH$ is a collection of all subsets
of $T$, i.e. $\calH = 2^T$. A bias $p\in [0,1]$ defines a measure
$\mu_p$ on $\calH$ such that $\mu_{p}(v) = 
{p}^{|v|}(1 - p)^{|T\setminus v|}$ for any $v \in \calH$. In our
construction we need several different biased measures defined
as follows.
For all $j = 1, \dots, r$, define $q_j := \frac{2j}{rk}$, and biases 
$p_j :=1
- q_j - \eps$. Each $p_j$ defines a biased measure $\mu_{p_j}$ over
a Long Code over any domain. Next,
we define the vertices of the hypergraph. 

\smallskip
{\bf Vertices.} We shall denote the set of vertices by $V$.
Consider a variable $x$ in the layer $X_l$ of the PCP.
For $i \in [k+1]$ and $j \in [r]$, let $\calH^x_{ij}$ be a
Long Code on the domain $R_{X_l}$ endowed with the bias $\mu_{p_j}$,
i.e.  $\mu_{p_j}(v) = {p_j}^{|v|}(1 - p_j)^{|R_{X_l}\setminus v|}$ for
all $v \in \calH^x_{ij} = 2^{R_{X_l}}$.  
The set of vertices corresponding to $x$ is $V[x] := 
\bigcup_{i=1}^{k+1} \bigcup_{j=1}^r \calH^x_{ij}$.
We define the weights on vertices to be proportional to its biased
measure in the corresponding Long Code. Formally, for any $v \in
\calH^x_{ij}$,
\begin{equation}
{\rm wt}(v) := \frac{\mu_{p_j}(v)}{L|X_l|r(k+1)}.
\end{equation}
The above conveniently ensures that for any $l \in [L]$,\
$\sum_{x \in X_l}{\rm wt}(V[x]) = 1/L$, and  
$\sum_{l \in [L]}\sum_{x \in X_l}{\rm wt}(V[x]) = 1$.
In addition to the vertices for each variable of the PCP, the instance
also contains $k+1$ \emph{dummy} vertices $d_1, \dots, d_{k+1}$ each
with a very large weight given by ${\rm wt}(d_i) := 2$ for $i\in[k+1]$. 
Clearly, this ensures that the total weight of all the vertices
in the hypergraph is $2(k+1) + 1$.  
As we shall see later, the edges shall be defined in such a way
that along with these weights would ensure that 
the maximum sized independent set shall contain all the dummy vertices. 
Before
defining the edges we define the $(k+1)$ partition $(V_1,\dots,
V_{k+1})$ of $V$ to be:
\begin{equation}
V_i = \left(\bigcup_{l=1}^L\bigcup_{x\in
X_l}\bigcup_{j=1}^r\calH^x_{ij}\right)\cup \{d_i\},
\end{equation}
for all $i=1,\dots, k+1$.
We now define
the hyperedges of the instance. In the rest of the section, the
vertices shall be thought of as subsets of their respective domains.

\smallskip
{\bf Hyperedges.} For every pair of variables $x$ and $y$ of the PCP 
such that there is a constraint $\pi_{x\rightarrow y}$, 
we construct edges as follows. 

\smallskip
(1.) Consider all permutations $\sigma : [k+1]\mapsto [k+1]$ and 
 sequences $(j_1,\dots, j_k, j_{k+1})$ such that, 
$j_1, \dots, j_k \in [r]\cup\{0\}$ and $j_{k+1} \in [r]$ such that:
$\sum_{i=1}^k\mathbbm{1}_{\{j_i \neq 0\}} \ q_{j_i} \geq 1$.

\smallskip
(2.) Add all possible hyperedges $e$ such that for all $i \in [k]$: 

\hspace{0.5cm}(2.a) If $j_i \neq 0$ then $e \cap V_{\sigma(i)} =: v_{\sigma(i)}
\in \calH^x_{\sigma(i), j_i}$, and,

\hspace{0.5cm}(2.b) If $j_i = 0$ then $e \cap V_{\sigma(i)} = d_{\sigma(i)}$
and,

\hspace{0.5cm}(2.c) $e \cap V_{\sigma(k+1)} =: u_{\sigma(k+1)}\in
\calH^y_{\sigma(k+1), j_{k+1}}$,

which satisfy,
\begin{equation} 
\pi_{x\rightarrow y} \left(\bigcap_{\substack{i:~i\in [k] \\ \ \
j_i\neq 0}}
v_{\sigma(i)} \right) \cap u_{\sigma(k+1)} = \emptyset. 
\label{eqn-intersect-empty}
\end{equation}
Let us denote the hypergraph constructed above by
$G(\Phi)$. From the construction it is clear the $G(\Phi)$ is $(k+1)$-partite
with partition $V = \cup_{i \in [k+1]}V_i$.

Note that the edges are defined in such a way that the set $\{d_1,
\dots, d_{k+1}\}$ is an independent set in the hypergraph. Moreover, since
the weight of each dummy vertex $d_i$ is $2$, while total weight of
all except the dummy vertices is $1$, this implies that any maximum
independent set $\calI$ contains all the dummy vertices. Thus,
$V\setminus \calI$ is a minimum vertex cover that does not 
contain any dummy vertices.
For convenience, the
analysis of our reduction, presented in the rest of this section,  
shall focus on the weight of 
$(\calI\cap V)\setminus \{d_1, \dots, d_{k+1}\}$. 

\subsection{Completeness}
In the completeness case, the instance $\Phi$ is a YES instance i.e.
there is a labeling $A$ which maps each variable $x$ in layer $X_l$ to
an assignment in $R_{X_l}$ for all $l=1, \dots, L$, such that all the
constraints of $\Phi$ are satisfied.

\noindent
Consider the set of vertices $\calI^*$ which satisfies the following
properties:\\
(1) $d_i \in \calI^*$ for all $i=1, \dots, k+1$.\\
(2) For all $l \in [L]$, $x \in X_l$, $i \in [k+1], j \in [r]$,
\begin{equation}
\calI^*\cap\calH^x_{ij} = \{v \in \calH^x_{ij} : A(x) \in v\}. 
\label{eqn-complete-IS}
\end{equation}
Suppose $x$ and $y$ are two variables in $\Phi$ with a constraint
$\pi_{x\rightarrow y}$ between them. Consider any $v \in \calI^*\cap
V[x]$ and $u \in \calI^*\cap V[y]$. The above construction of $\calI^*$
along with the fact that the labeling $A$ satisfies the constraint
$\pi_{x\rightarrow y}$
implies that $A(x) \in v$ and $A(y)\in u$ and
$A(y) \in \pi_{x\rightarrow y}(v) \cap u$. Therefore, Equation 
\eqref{eqn-intersect-empty} of the 
construction is not satisfied by the vertices in $\calI^*$, and so 
$\calI^*$ is an independent set in the hypergraph.
By Equation \eqref{eqn-complete-IS}, 
the fraction of the weight of the Long Code $\calH^x_{ij}$ which
lies in $\calI^*$ is $p_j$, for any variable $x$, $i\in [k+1]$ 
and $j\in [r]$. Therefore,
\begin{equation}
\frac{{\rm wt}(\calI^*\cap V[x])}{{\rm wt}(V[x])} =
\frac{1}{r}\sum_{j=1}^r p_j  = 1 - \frac{1}{k}\left(1 +
\frac{1}{r}\right) - \eps,
\end{equation}
by our setting of $p_j$ in Section \ref{sec:construction}.
The above yields that 
\begin{equation}
{\rm wt}\left(\calI^* \cap (V\setminus \{d_1,\dots,
d_{k+1}\})\right) = 1 - \frac{1}{k}\left(1 +
\frac{1}{r}\right) - \eps \geq 1 - \frac{1}{k} - 2\eps, 
\label{eq-completeness}
\end{equation}
for a small enough value of $\eps > 0$ and our setting of the parameter
$r$.

\subsection{Soundness}
For the soundness analysis we have that $\Phi$ is a NO instance 
as given in Theorem \ref{thm-multi} and we wish to prove that the size
of the maximum independent set in $G(\Phi)$ is appropriately 
small. 
For a contradiction, we assume that there is a maximum
independent set $\calI$ in $G(\Phi)$ such that,
\begin{equation}
{\rm wt}(\calI \cap (V\setminus \{d_1,\dots,
d_{k+1}\})) \geq 1 - \frac{k}{2(k+1)} + \eps.
\end{equation}
Define the set of variables $X'$ to be as follows:
\begin{equation}
X' := \left\{x\textnormal{ a variable in }\Phi : \frac{{\rm
wt}(\calI\cap V[x])}{{\rm wt}(V[x])} \geq 1 - \frac{k}{2(k+1)} +
\frac{\eps}{2}\right\}. \label{eq-defX'}
\end{equation}
An averaging argument shows that ${\rm wt}(\cup_{x\in X'}V[x]) \geq
\eps/2$. A further averaging implies that there are $\frac{\eps}{4}L =
\frac{8}{\eps}$ layers of $\Phi$ such that $\frac{\eps}{4}$ fraction
of the variables in each of these layers belong to $X'$. Applying the
Weak Density property of $\Phi$ given by Definition \ref{def-weakly-dense}
and Theorem \ref{thm-multi}
yields two layers $X_{l'}$ and $X_{l''}$  ($l' < l''$) such that
$\frac{\eps^2}{64}$ fraction of the constraints between them are
between variables in $X'$. The rest of the analysis shall focus on
these two layers and for convenience we shall denote $X'\cap X_{l'}$ by $X$ 
and $X' \cap X_{l''}$ by $Y$, and denote the respective label sets by
$R_X$ and $R_Y$.

Consider any variable $x \in X$. For any $i\in[k+1], j\in[r]$,
call a Long Code $\calH^x_{ij}$ \emph{significant} if 
$\mu_{p_j}(\calI\cap \calH^x_{ij}) \geq
\frac{\eps}{2}$. From Equation
\eqref{eq-defX'} and an averaging argument we obtain that,
\begin{equation}
\left|\{(i,j) \in [k+1]\times[r] :
\calH^x_{ij}\textnormal{  is \emph{significant}.}\}\right|
\geq \left(1 - \frac{k}{2(k+1)}\right)(r(k+1)) = \frac{rk}{2} +r. 
\label{eq-significant-bd}
\end{equation}
Using an
analogous argument we obtain a similar statement for every variable $y
\in Y$ and corresponding Long Codes $\calH^y_{ij}$.
The following structural lemma follows from the above bound.
\begin{Lem} \label{lem-exist-good-sequence}
Consider any variable $x \in X$. Then there exists a sequence $(j_1,
\dots, j_{k+1})$ with $j_i \in [r]\cup\{0\}$ for $i\in [k+1]$; 
such that the Long Codes
$\{\calH^x_{i, j_i} \mid i\in [k+1] \textnormal{ where }j_i \neq 0\}$, 
are all \emph{significant}. Moreover,
\begin{equation}
\sum_{i=1}^{k+1}j_i \geq \frac{rk}{2} + r\ .
\end{equation}
\end{Lem}
\begin{proof}
For all $i \in [k+1]$ choose $j_i$ as follows: if none of the Long
Codes $\calH^x_{ij}$ for $j\in[r]$ are \emph{significant} then let $j_i
:= 0$, otherwise let $j_i := \max \{j \in
[r] : \calH^x_{ij}\textnormal{ is \emph{significant}}\}$. It is easy
to see that $j_i$ is an upper bound on the number of significant Long
Codes of the form $\calH^x_{ij}$. Therefore,
\begin{equation}
\sum_{i=1}^{k+1}j_i \ \geq \ \left|\{(i,j) \in [k+1]\times[r] :
\calH^x_{ij}\textnormal{  is \emph{significant}.}\}\right|
\ \geq \ \frac{rk}{2} + r \ \ \ \ \textnormal{(From Equation
\eqref{eq-significant-bd})} \label{eq-sequence-bd} 
\end{equation}
which proves the lemma.
\end{proof}
Next we define the decoding procedure to define a label for any
given variable $x \in X$.

\subsubsection{Labeling for variable $x \in X$} 
The label $A(x)$ for each variable $x
\in X$ is chosen independently via
the following three step (randomized) procedure. 

\smallskip
Step 1. Choose a sequence $(j_1, \dots,j_{k+1})$ yielded by
Lemma \ref{lem-exist-good-sequence} applied to $x$.
 
\smallskip
Step 2. Choose an element $i_0$ uniformly at random from $[k+1]$.

\smallskip
Before describing the third step of the procedure we require the
following lemma.
\begin{Lem}\label{lem-vertices-choosing}There exist vertices $v_i \in
\calI\cap\calH^x_{ij_i}$ for every $i : i \in [k+1]\setminus\{i_0\}, j_i
\neq 0$, and an integer $t := t(\eps)$ satisfying:
\begin{equation}
\left|\bigcap_{\substack{i : i \in [k+1]\setminus\{i_0\}, \\ j_i \neq
0}}v_i\right| < t.\label{eq-vertices-intersection}
\end{equation}
\end{Lem}
\begin{proof} Since $j_{i_0} \leq r$ it is easy to see,
\begin{equation}
\sum_{i\in [k+1]\setminus\{i_0\}} j_i \geq \frac{rk}{2} 
\ \Rightarrow \ \sum_{\substack{i : i\in[k+1]\setminus\{i_0\}, \\ j_i
\neq 0}} q_{j_i} \geq 1. \label{eq-sj-bd}
\end{equation}
Moreover, since the sequence $(j_1, \dots, j_{k+1})$ was obtained by
Lemma \ref{lem-exist-good-sequence} applied to $x$, we know that
$\mu_{p_{j_i}}(\calI\cap\calH^x_{ij_i}) \geq \frac{\eps}{2}$,
$\forall i : i\in[k+1]\setminus\{i_0\}, j_i\neq 0$.
Combining this with Equation \eqref{eq-sj-bd} and
Lemma \ref{lem:intersecting-measure} we
obtain that for some integer $t := t(\eps)$ the collection
 of set families
$\{\calH^x_{ij_i} :  i\in[k+1]\setminus\{i_0\}, j_i\neq 0\}$ is not
$k'$-wise $t$-cross-intersecting, where $k' = |\{i\in[k+1]\setminus\{i_0\} 
:j_i\neq 0\}|$. 
This proves the lemma.
\end{proof}
The third step of the labeling procedure is as follows:

\smallskip
Step 3. Apply Lemma \ref{lem-vertices-choosing} to obtain the 
the vertices $v_i \in \calI\cap\calH^x_{ij_i}$ for every 
$i : i \in [k+1]\setminus\{i_0\}, j_i \neq 0$ satisfying Equation
\eqref{eq-vertices-intersection}. Define $B(x)$ as,
\begin{equation}
B(x) := \bigcap_{\substack{i : i \in [k+1]\setminus\{i_0\}, \\ j_i \neq
0}}v_i, \label{eq-Bx}
\end{equation}
noting that $|B(x)| < t$. Assign a random label from $B(x)$ to the
variable $x$ and call the assigned label $A(x)$.

\subsubsection{Labeling for variable $y \in Y$}
After labeling the variables $x \in X$ via the procedure above, 
we construct a labeling  $A(y)$ for any variable 
$y \in Y$ by defining,
\begin{equation}
A(y) := \textnormal{argmax}_{a \in R_Y}\left|\{x \in X\cap N(y)\ \mid\
a \in \pi_{x\rightarrow y}(B(x))\}\right|,
\end{equation}
where $N(y)$ is the set of all variables that have a constraint with
$y$. The above process selects a label for $y$ which lies in maximum
number of projections of $B(x)$ for variables $x \in X$ which have a
constraint with $y$.

The rest of this section is devoted to lower bounding the number of
constraints satisfied by the labeling process, and thus obtain a
contradiction to the fact that $\Phi$ is a NO instance.

\subsubsection{Lower bounding the number of satisfied constraints}
Fix a variable $y \in Y$. Let $U(y) := X\cap N(y)$, i.e. the variables
in $X$ which have a constraint with $y$.
Further, define the set $P(y) \subseteq [k+1]$ as follows,
\begin{equation}
P(y) = \{i \in [k+1]\ \mid\ \exists j\in[r] \textnormal{ such that }
\mu_{p_j}(\calI\cap \calH^y_{ij}) \geq \eps/2\}.
\end{equation}
In other words, $P(y)$ is the set of all those 
indices in $[k+1]$ such that
there is a \emph{significant} Long Code corresponding to each of them.
Applying Equation \eqref{eq-significant-bd} to $y$ we obtain that
there at least $\frac{r(k+2)}{2}$ \emph{significant} Long Codes
corresponding to $y$, and therefore $|P(y)| \geq \frac{k+1}{2} \geq 1$. 
Next
we define subsets of $U(y)$ depending on the outcome of Step 2 in the
labeling procedure for variables $x \in U(y)$.
For $i \in [k+1]$ define, 
\begin{equation}
U(i, y) := \{x \in U(y)\ \mid\ i\textnormal{ was chosen in Step 2 of
the labeling procedure for } x\},
\end{equation}
and,
\begin{equation}
U^*(y) := \bigcup_{i \in P(y)} U(i, y).
\end{equation}
Note that $\{U(i, y)\}_{i\in[k+1]}$ is a partition of $U(y)$. Also,
since $|P(y)| \geq \frac{k+1}{2}$ and the labeling procedure for each
variable $x$ chooses the index in Step 2 uniformly and independently
at random we have,
\begin{equation}
\mathbb{E}[|U^*(y)|] \geq \frac{|U(y)|}{2}, \label{eq-expect-U}
\end{equation} 
where the expectation is over the random choice of the indices in Step
2 of the labeling procedure for all $x \in U(y)$. Before continuing we
need the following simple lemma (proved as Claim 5.4  in
\cite{DGKR03}). 
\begin{Lem}\label{lem-disjoint}
Let $A_1,\dots, A_N$ be a collection of $N$ sets, each of size at most
$T\geq 1$. If there are not more than $D$ pairwise disjoint sets in
the collection, then there is an element that is contained in at least
$\frac{N}{TD}$ sets.
\end{Lem}
Now consider any $i' \in P(y)$ such that $U(i', y)\neq \emptyset$ 
and a variable $x \in
U(i', y)$. Since $i' \in
P(y)$ there is a \emph{significant} Long Code $\calH^y_{i'j'}$ for some
$j' \in [r]$. Furthermore, since $\calI$ is an independent set there 
cannot be a $u
\in \calI\cap\calH^y_{i',j'}$ such that $\pi_{x\rightarrow y}(B(x))\cap u 
= \emptyset$,
otherwise the following set of $k+1$ vertices, 
$$\{v_i\ \mid\ i\in[k+1]\setminus\{i'\}, j_i\neq 0\} \cup \{d_i \ \mid\  
i\in[k+1]\setminus\{i'\}, j_i = 0\}\cup\{u\}$$
form an edge in $\calI$, where $v_i, j_i$ ($i \in [k+1]$) are as
constructed in the labeling procedure for $x$.

Consider the collection of 
sets $\pi_{x\rightarrow y}(B(x))$ for all $x \in U(i',
y)$. Clearly each set is of size less than $t$. 
Let $D$ be the maximum number of disjoint sets in this collection.
Each disjoint set independently reduces the measure of
$\calI\cap\calH^y_{i',j'}$ by a factor of $(1 - (1 - p_{j'})^t)$.
However, since $\mu_{p_{j'}}(\calI\cap\calH^y_{i',j'})$ is at least
$\frac{\eps}{2}$, this implies that $D$ is at most
$\log(\frac{\eps}{2})/\log(1-(2/rk)^t)$, since $p_{j'} \leq 1 -
\frac{2}{rk}$. Moreover, since $t$ and $r$ depends only on $\eps$,  
the upper bound on $D$ also depends only on $\eps$.

Therefore by Lemma \ref{lem-disjoint}, 
there is an element $a \in R_Y$ such that $a \in
\pi_{x\rightarrow y}(B(x))$ for at least $\frac{1}{Dt}$ fraction of $x
\in U(i',y)$. Noting that this bound is independent of $j'$ and that
$\{U(i',y)\}_{i' \in P(y)}$ is a partition of $U^*(y)$, we obtain that
there is an element $a \in R_Y$ such that $a \in \pi_{x\rightarrow
y}(B(x))$ for $\frac{1}{(k+1)Dt}$ fraction of $x \in U^*(y)$.
Therefore, in Step 3 of the labeling procedure
when a label $A(x)$ is chosen uniformly at random from
$B(x)$, in exception, $a = \pi_{x\rightarrow y}((A(x))$ for
$\frac{1}{(k+1)Dt^2}$ fraction of $x \in U^*(y)$. Combining this with
Equation \eqref{eq-expect-U} gives us that there is a labeling to the
variables in $X$ and $Y$ which satisfies $\frac{1}{2(k+1)Dt^2}$ fraction
of the constraints between variables in $X$ and $Y$ which is in turn
at least $\frac{\eps^2}{64}$ fraction of the constraints between the
layers $X_{l'}$ and $X_{l''}$. Since $D$ and $t$ depend only on $\eps$,
choosing the parameter $R$ of $\Phi$ to be large enough we obtain a
contradiction to our supposition on the lower bound on the size of the
independent set. Therefore in the Soundness case, any for any
independent set $\calI$, 
$${\rm
wt}(\calI\cap(V\setminus\{d_1,\dots,d_{k+1}\})) \leq 1 -
\frac{k}{2(k+1)}
+ \eps.$$
Combining the above with Equation \eqref{eq-completeness} of 
the analysis in the Completeness case yields
a factor $\frac{k^2}{2(k+1)} - \delta$ (for any $\delta > 0$) 
hardness for approximating $(k+1)$-{\sc Hyp\-VC\--Par\-tite} .

Thus, we obtain a factor $\frac{k}{2} - 1 + \frac{1}{2k} -\delta$
hardness for approximating \hvcpartite.

%%% Local Variables: 
%%% mode: latex
%%% TeX-master: "vc"
%%% End: 